\documentclass[lettersize,journal]{IEEEtran}
\IEEEoverridecommandlockouts
\usepackage{cite}
\usepackage{amsmath,amssymb,amsfonts}

\usepackage{algorithmic}
\usepackage{graphicx}
\usepackage{textcomp}
\usepackage{xcolor}
\usepackage{caption}
\usepackage{subcaption}
\usepackage{multirow}
\usepackage{tablefootnote}
\usepackage{amsthm}
\newtheorem{theorem}{Theorem}
\newtheorem{lemma}{Lemma}
\usepackage{cite}
\usepackage{amsfonts}
\usepackage{amssymb}
\usepackage[ruled,vlined]{algorithm2e}
\usepackage{comment}
\usepackage{empheq} 
\usepackage{cases}

\usepackage[T1]{fontenc}

\def\BibTeX{{\rm B\kern-.05em{\sc i\kern-.025em b}\kern-.08em
    T\kern-.1667em\lower.7ex\hbox{E}\kern-.125emX}}

\begin{document}

\title{Trace-Distance based End-to-End Entanglement Fidelity with Information Preservation in Quantum Networks}

\author{Pankaj~Kumar,~\IEEEmembership{Student~Member,~IEEE,}
       Binayak~Kar,~\IEEEmembership{Member,~IEEE,}
    and~Shan-Hsiang~Shen,~\IEEEmembership{Member,~IEEE}
\thanks{This work was supported by the National Science and Technology Council (NSTC), Taiwan, under Grants 112-2221-E-011-057 and 113-2221-E-011-152 (\textit{Corresponding author: Binayak Kar.})}
\thanks{P. Kumar, B. Kar, and S.-H. Shen are with the Department of Computer Science and Information Engineering, National Taiwan University of Science and Technology, Taipei, Taiwan. (e-mail: pnkazaayan@gmail.com, bkar@mail.ntust.edu.tw, and sshen@csie.ntust.edu.tw).
}}

\markboth{}%
{Shell \MakeLowercase{\textit{et al.}}: A Sample Article Using IEEEtran.cls for IEEE Journals}


\maketitle

\begin{abstract}
Quantum networks hold the potential to revolutionize a variety of fields by surpassing the capabilities of their classical counterparts. 
Many of these applications necessitate the sharing of high-fidelity entangled pairs among communicating parties. However, the inherent nature of entanglement leads to an exponential decrease in fidelity as the distance between quantum nodes increases. This phenomenon makes it challenging to generate high-fidelity entangled pairs and preserve information in quantum networks. To tackle this problem, we utilized two strategies to ensure high-fidelity entangled pairs and information preservation within a quantum network. First, we use closeness centrality as a metric to identify the closest nodes in the network. 
Second, we introduced the trace-distance based path purification (TDPP) algorithm, specifically designed to enable information preservation and path purification entanglement routing. 
This algorithm identifies the shortest path within quantum networks using closeness centrality and integrates trace-distance computations for distinguishing quantum states and maintaining end-to-end (E2E) entanglement fidelity. Simulation results demonstrate that the proposed algorithm improves network throughput and E2E fidelity while preserving information compared to existing methods.
\end{abstract}

\begin{IEEEkeywords}
Entanglement fidelity, information preservation, trace-distance, quantum routing.
\end{IEEEkeywords}

\section{Introduction} \label{intro}
Quantum computing, a burgeoning field that leverages principles of quantum physics, holds the potential to solve extraordinarily complex problems beyond the reach of traditional computers \cite{grover1996fast}. In specific problem domains, it has the potential to surpass conventional computing methods. One notable application is quantum computing's potential to resolve NP-hard problems such as leader selection and integer factorization in polynomial time \cite{shor1994algorithms}. The growth of quantum internet infrastructure, which supports crucial communication protocols, facilitates these capabilities \cite{bugalho2023distributing}. 

It is anticipated that each quantum computer will be able to handle only a certain number of quantum bits (qubits) in the future. Small quantum computers can be interconnected through a quantum network to address this issue, forming a distributed processing system \cite{li2021effective}. A quantum network links quantum nodes via optical connections, including quantum processors and repeaters. These nodes can generate, store, exchange, and process quantum information. When two quantum nodes act as the source and destination (S-D) pairs for information exchange, the quantum network establishes an entanglement connection between them. 
Figure \ref{fig.fig1} illustrates the generation of an entanglement link from the source to a destination node. In the network, each quantum node is occupied with either $\rho_i$ or $\sigma_i$, the density matrix of the quantum state as shown in Figure  \ref{fig.fig1}(a). After performing the swapping operation at adjacent nodes as shown in Figure  \ref{fig.fig1}(b), the quantum state changes to a maximally mixed state at the destination nodes. With the aid of a quantum repeater, quantum nodes can be linked over greater distances and perform entanglement-swapping operations. Subsequently, information is transmitted from the source to the destination through qubits over the potentially noisy channel. This process involves performing an Einstein–Podolsky–Rosen (EPR) pair state measurement at the local node, which is supported by classical communication \cite{zhao2021redundant}.

\begin{figure}[!t]
    \centering
    \begin{subfigure}[b]{0.40\textwidth}
    \centering
    \includegraphics[width=\textwidth]{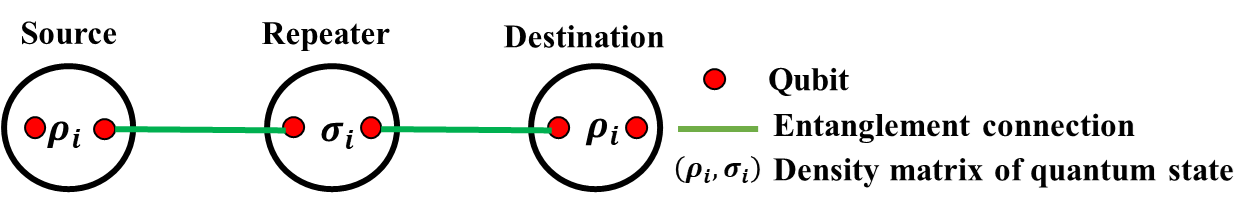}
    \caption{Before entanglement connection}
    \end{subfigure}
    \begin{subfigure}[b]{0.40\textwidth}
    \centering
    \includegraphics[width=\textwidth]{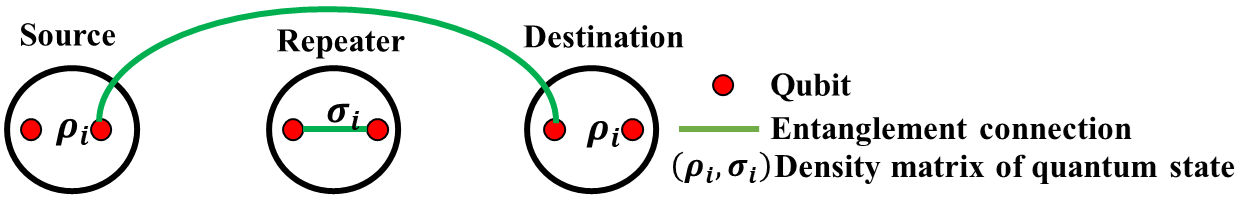}
    \caption{After entanglement connection}
    \end{subfigure}
     \caption{llustration of information transmission in a quantum network, where $\rho_i$ is the density matrix of quantum state at the source node and $\sigma_i$ is the density matrix of the quantum state received at adjacent node respectively.}
    \label{fig.fig1}
\end{figure}

Reliable long-distance quantum information transmission is a fundamental requirement for many quantum applications. Many existing works \cite{caleffi2020quantum,pirandola2019end,pant2019routing,patil2022entanglement,caleffi2017optimal,santos2023shortest} focus on long-distance quantum entanglement routing protocols. These protocols aim to establish end-to-end entanglement through quantum repeaters via entanglement swapping \cite{zukowski1993event} in a quantum network, improving network throughput, enhancing robustness, and serving more users. However, due to the fragile nature of quantum information, qubits are susceptible to decoherence \cite{schlosshauer2019quantum} through interactions with the environment. For example, the generated entangled pairs may not be perfectly entangled, and attenuation in physical links and imperfect swapping operations can lead to corruption during the establishment of long-distance entanglement. As a result, the established end-to-end entanglement may not be at the desired states and cannot be used for reliable quantum information transmission. Typically, fidelity \cite{li2022fidelity} is used to quantify the quality of an entanglement link. The value of fidelity ranges from 0 to 1 and measures how well a quantum channel preserves quantum information. Despite some recent works \cite{cacciapuoti2019quantum,zhao2022e2e} considering fidelity guarantees when designing entanglement routing protocols, it remains essential to explicitly verify the quality of entanglement links before transmitting important quantum information.

Closeness centrality enhances entanglement distribution in quantum networks by identifying nodes with the shortest paths to all other nodes, reducing communication delays, and optimizing resource allocation. High-centrality nodes serve as hubs, minimizing the need for intermediate entanglement swaps and lowering the risk of errors or decoherence. 
This results in more reliable and efficient paths for distributing entanglement, helping preserve the fidelity of quantum states.
By optimizing path selection, closeness centrality enhances the network’s resilience and speed, enabling more reliable communication and efficient resource management across the quantum network. However, an important metric for evaluating the integrity and quality of quantum information is the trace-distance between quantum states, which is rarely considered in existing entanglement routing designs. Quantum systems are inherently fragile and susceptible to decoherence \cite{schlosshauer2019quantum} and noise \cite{mor2024influence}, factors that can degrade the quality of quantum information. One can quantify the impact of these disturbances by comparing the trace-distance between the initial state and the received state. A low trace-distance signifies that the quantum state has been well-preserved, whereas a high trace-distance indicates significant alteration, highlighting areas where noise mitigation strategies are required.

Given these considerations, this paper focuses on handling noisy environments and maintaining end-to-end (E2E) fidelity within a general quantum network \cite{jia2021improved}. Additionally, we emphasize preserving information during quantum routing in the network. 
To achieve this, we propose a trace-distance based path purification (TDPP) routing algorithm that aims to achieve multiple objectives \cite{martins1984multicriteria}, including closeness centrality, information preservation, and E2E entanglement fidelity.
The primary objective is to identify entanglement paths from source to destination based on the closeness centrality of the node and perform trace-distance calculations to preserve information and make purification decisions. This involves analyzing the characteristics of the purification operation and proposing optimal decisions to maintain both information and E2E entanglement fidelity during quantum routing. While existing research addresses entanglement distribution design and minimum E2E fidelity requirements \cite{elliott2002building,shi2020concurrent}, it often overlooks trace-distance based purification, limiting improvements in the fidelity of individual EPR pairs. Our approach optimizes entanglement fidelity in routing, preserving information through strategic purification decisions, and leveraging purification operations to achieve improved E2E fidelity even with low-fidelity EPR pairs.
The novel contributions of our work are as follows:

\begin{enumerate}
    \item{We propose a trace-distance based approach that guarantees information preservation and reliable information flow in quantum networks involving multiple S-D pairs.}
    \item {To meet the demand for high-quality entanglement connections across various quantum applications, we introduce the first trace-distance based path purification (TDPP) technique. This architecture guarantees fidelity in E2E connections for S-D pairs within quantum networks.}
    \item{We conducted extensive simulations to showcase our algorithm's superior performance compared to other existing algorithms.}
\end{enumerate}

 The rest of the paper is organized as follows. Section \ref{related} discusses related works. Section \ref{NMCD} presents the network model, quantum properties, and quantum operations. Section \ref{QND} details the path selection, routing algorithm, and methods to preserve information based on the quantum state's static and dynamic measures. Performance evaluation analysis is described in Section \ref{PA}. Finally, Section \ref{conclusion} concludes the paper.


\begin{table*}[!t]
\centering
\caption{Related Works}
\begin{tabular}{l c c c c c c c c l}
\hline \hline
\rotatebox{0}{References} & \rotatebox{0}{Fidelity}   & \rotatebox{0}{Quantum} & \rotatebox{0}  {Quantum} & \rotatebox{0} {Entanglement} & \rotatebox{0}{Throughput}  & \rotatebox{0} {Static}   &  \rotatebox{0} {Dynamic} & \rotatebox{0}{Information} & \rotatebox{0}{Approaches}  \\ 
 &    &  \rotatebox{0} {Memory}   &  \rotatebox{0}{Channel}  &  \rotatebox{0}{Fidelity}  &   &  \rotatebox{0}{Measurement}    &  \rotatebox{0}{Measurement}  & \rotatebox{0} {Preservation} &   \\
 &    &    &  \rotatebox{0} {Capacity} &  \rotatebox{0}{Purification} &   &  \rotatebox{0}{of State}   &  \rotatebox{0}{of State} & &   \\ \hline
\cite{pant2019routing} & $\checkmark$ & $\times$ & $\times$ & $\times$ & $\times$ & $\times$ & $\times$ & $\times$ & Shortest-Path \\  
\cite{li2021effective} & $\checkmark$ & $\checkmark$ & $\times$ & $\times$ & $\times$ & $\times$ & $\times$ & $\times$ & ELP \\  
\cite{shi2020concurrent} & $\times$ & $\checkmark$ & $\checkmark$ & $\times$ & $\times$ & $\times$ & $\times$  & $\times$ & Q-CAST \\ 
\cite{zhao2021redundant} & $\times$ & $\checkmark$ & $\checkmark$ & $\times$ & $\checkmark$ & $\times$ & $\times$ & $\times$ & REPS\\  
\cite{zhao2022e2e} & $\checkmark$ & $\checkmark$ & $\checkmark$ & $\checkmark$ & $\checkmark$ & $\times$ & $\times$ & $\times$ & EFiARP\\ 
\cite{li2022fidelity} & $\checkmark$ & $\checkmark$ & $\checkmark$ & $\checkmark$ & $\checkmark$ & $\times$ &$\times$ & $\times$ & Q-PATH\\  
Our & $\checkmark$ & $\checkmark$ & $\checkmark$ & $\checkmark$ & $\checkmark$ & $\checkmark$ & $\checkmark$ & $\checkmark$ & TDPP \\  \hline \hline
\end{tabular}
\label{tab:related}\\
\end{table*}

\section{Related Works} \label{related}
Quantum network routing has been extensively studied \cite{yin2017satellite}. 
Previous work has focused on several key themes that are shown in Table \ref{tab:related}, including entanglement generation, fidelity purification, and network design protocols \cite{pant2019routing, kar2023routing}. This section provides a comprehensive review of relevant studies in quantum network routing and highlights the gaps our research aims to address.

Li et al. \cite{li2021effective} propose an efficient entanglement link path (ELP) protocol for remote entanglement generation in quantum networks. The ELP protocol utilizes a modularized and flexible routing scheme based on a tree structure to distribute entangled qubit pairs effectively across the network. It is evaluated in terms of network topologies, entanglement resources, and error models, demonstrating its suitability for quantum repeater networks, quantum key distribution, and distributed quantum computing. Mihir et al. \cite{pant2019routing} focus on developing the quantum internet, emphasizing the establishment of entangled qubit pairs between distant nodes. While their work primarily focuses on efficient entanglement generation within available resources, they propose protocols for quantum repeater nodes that leverage multiple paths to achieve higher entanglement rates compared to linear chains.

Shi et al. \cite{shi2020concurrent} explored concurrent entanglement routing in quantum networks. They proposed both source-concurrent and target-concurrent designs to optimize resource utilization and enhance robustness. Zhao et al. \cite{zhao2021redundant} introduced the redundant entanglement provisioning and selection (REPS) scheme, which provisions and selects redundant entangled pairs to enhance throughput and robustness. REPS adapts to network conditions by considering topology, fidelity, and resource availability. It provides backup resources for failure tolerance and flexible selection of successfully created entanglement links, contributing to efficient and reliable routing in quantum networks.

Zhao et al. \cite{zhao2022e2e} extended their previous work by introducing an E2E fidelity-aware routing and purification (EFiARP) algorithm to optimize high-fidelity entanglement distribution in quantum networks. They addressed limited throughput caused by noise and errors by considering entanglement fidelity in the optimization of routing and purification. EFiARP was designed to prepare candidate entanglement paths, determine optimal purification schemes, and select paths to maximize network throughput within resource constraints. Meanwhile, Li et al. \cite{li2022fidelity} proposed the fidelity guaranteed entanglement routing (Q-PATH) algorithm, ensuring a minimum fidelity for entangled pairs. Q-PATH utilized an iterative routing algorithm for single source-destination pairs and a greedy-based algorithm for multiple pairs, with considerations for resource allocation and re-routing.

In this study, we tackle the challenges associated with E2E fidelity and information preservation between S-D pairs in quantum networks when establishing connections on quantum links. 
As such, we introduce a novel routing metric based on trace-distance tailored for quantum networks, rather than relying on hop count {\cite{santos2023shortest, hahn2019quantum}}, as conventional approaches do.

\section{Background of Quantum Network and Communication} \label{NMCD}
This section encompasses two key components: the motivation behind our research and our approach to network model design. We delve into the concepts of state closeness i.e., trace-distance, dynamic measurement of a quantum state, and entanglement routing problems in quantum networks, offering an analysis of their properties.

\subsection{Network Topology and Notations}
A quantum network, denoted as $G$ = ($\mathcal{V}$, $E$), is composed
of a set of nodes $\mathcal{V}$ and a set of edges $E$, which represent
quantum entities interconnected by shared quantum channels.
For a given quantum network, each node $v$ $\in$ $\mathcal{V}$ is equipped with quantum memory.
Each edge $(u,v) \in E$ has a channel capacity $c_{u,v}$, which is defined as the number of EPR pairs shared between the adjacent quantum nodes $u$ and $v$.
Let $s_k$ and $d_k$ be the source node and destination node for the $k^{th}$ source-destination pair $(s_k,d_k)$, respectively.
Let $P^{SPF}_{minCost}$ be the set of shortest paths, and $P_i(s_k,d_k)$ be the $i^{th}$ shortest path for the source-destination pair $(s_k, d_k)$.
In the network, let $m_u$ represent the memory size at node $u$, and $(p_i, q_i)$ denote a probability distribution of the quantum state.
The trace-distance of the quantum state over an edge ($u, v$) is denoted as $D_{u,v}$($\rho$,\ $\sigma$), and the maximum trace-distance of an edge in the selected path from source $s_{k}$ to destination $d_{k}$ is $D_{s_k,d_k}^{max}$($\rho$,\ $\sigma$). The closeness centrality of a node $v$ is represented by $C_{v}$.
The quantum states of the node are represented by $\left|\psi\right\rangle$ and $\left|\varphi\right\rangle$.
The density matrices of quantum states $\left|\psi\right\rangle$ and $\left|\varphi\right\rangle$ are denoted as $\rho$ and $\sigma$, respectively, 
where $U$ and $V$ are unitary operators that perform quantum operations at the nodes, and $Z$ is the Pauli phase operator.
Let $\varepsilon$ be the quantum channel and $\rho_{\varepsilon}$ be the density matrix associated with the quantum channel.
The operation of the quantum channel on a state is represented by $\zeta(\left|\psi\right\rangle, \left\langle\psi|\right)$.
The fidelity over an edge $(u,v)$ is denoted as $F_{u,v}$, and the minimum fidelity of a quantum channel is $F_{min}\left(\varepsilon\right)$.
The fidelity of the quantum state over an edge is $F_{u,v}$($\rho$,\ $\sigma$).
The maximum fidelity of an edge in the selected path from source $s_{k}$ to destination $d_{k}$ is ${\hat{F}}_{s_k,d_k}^{SelEdge}$($\rho$,\ $\sigma$), and the purified fidelity of the edge with the maximum trace distance in the selected path from source $s_{k}$ to destination $d_{k}$ is $F_{D_{s_k,d_k}^{max}}^{purific}$($\rho$,\ $\sigma$).
The terminologies used in this paper are summarized in Table~\ref{tab:table1}.

\begin{table}[!t]
    \centering
    \caption{List of commonly used Variables and Notations}
    \label{tab:table1}
    \begin{tabular}{|l|l|}
         
         \hline 
          \textbf{Symbol} & \textbf{Description} \\  \hline 
          & \underline{Networks} \\
         $\left(s_k,d_k\right)$ & $k^{th}$ source-destination (S-D) pair \\
          $s_k$ & Source node of $k^{th}$ S-D pair \\
          $d_k$ & Destination node of $k^{th}$ S-D pair \\
         $P_{i}(s_k,d_k)$ & $i^{th}$ Path of S-D pairs $(s_k,d_k)$ \\
          $P_{minCost}^{SPF}$ & Set of shortest paths \\
          \textit{$c_{u,v}$} & Channel capacity over an edge ($u, v$)\\
          \textit{$m_{u}$} &Quantum memory size at node $u$ \\
          $C_{v}$ & Closeness centrality of node  \\
           \hline
           & \underline{Probability and Matrix} \\
          \textit{$p_i$, $q_i$} & Probability distribution of quantum state  \\
        $\rho$, $\sigma$ & Density matrix  \\
        $\rho_{\varepsilon}$ & Density matrix of quantum channel \\
         \hline
         & \underline{State, Channel and Operators} \\
        $\left|\psi\right\rangle$, $\left|\varphi\right\rangle$ & Quantum state of node  \\
         $\left|m\right\rangle$ & Arbitrary quantum state \\
         $\varepsilon$ & Quantum channel \\
          $\zeta(\left|\psi\right\rangle, \left\langle\psi|\right)$ & Quantum operation on a channel \\
        $U$, $V$ & Unitary operators \\
         \textit{Z} & Pauli phase operator \\
         \hline
         & \underline{Distance} \\
          $D_{u,v}$($\rho$,\ $\sigma$) & Trace-distance of quantum state over an edge (\textit{u},\textit{v}) \\
        $D_{s_k,d_k}^{max}$($\rho$,\ $\sigma$) & Maximum trace-distance of an edge in the selected \\ & path from $s_k$ to $d_k$ \\
         \hline
         & \underline{Fidelity} \\
       $F_{u,v}$ & Fidelity over an edge ($u, v$) \\
       $F_{min}\left(\varepsilon\right)$ & Minimum fidelity of quantum channel $\varepsilon$ \\
         $F_{u,v}$($\rho$,\ $\sigma$) & Fidelity of the quantum state over an edge (\textit{u},\textit{v})  \\
         $\hat{F}(\rho,\varepsilon(\rho))$ & Fidelity of quantum state  \\
         ${\hat{F}}_{s_k,d_k}^{SelEdge}$($\rho$,\ $\sigma$) & Maximum fidelity of an edge in the selected path \\ & from $s_k$ to $d_k$ \\
         $F_{D_{s_k,d_k}^{max}}^{purific}$($\rho$,\ $\sigma$) & Purified fidelity of the edge with maximum trace- \\ & distance in the selected path from $s_k$ to $d_k$ \\
          \hline 
    \end{tabular}
    
\end{table}

\emph{1) Quantum nodes:} Every quantum node serves as a processing unit for quantum information. Each node possesses the capability to generate, store, transmit, and manipulate quantum states. The establishment of link-level entanglement between any two neighboring quantum nodes is facilitated through the process of entanglement generation. Each quantum node is furnished with a finite set of quantum memory units \cite{li2020hyperfine} and the requisite hardware for executing quantum operations, such as entanglement swapping. Quantum nodes encompass various entities, primarily quantum end nodes, quantum repeaters, and quantum routers. Quantum end nodes play a pivotal role in teleporting unknown quantum bits and processing quantum information to support diverse quantum applications \cite{cacciapuoti2020entanglement}. Quantum repeaters, on the other hand, focus on extending the reach of entanglement distribution. Acting as networking devices, quantum routers interconnect numerous quantum end nodes and direct each entanglement routing request to the intended destination node. The collaborative efforts of quantum routers and quantum repeaters facilitate the seamless interconnection of multiple quantum end nodes within a quantum network.

\emph{2) Quantum repeaters:} In quantum networks, establishing connections between non-adjacent nodes is a challenging task that necessitates using quantum repeaters \cite{van2008system} for long-distance entanglement sharing via quantum swapping. Quantum repeaters connect with other repeaters and nodes through the classical internet, enabling information exchange over a channel \cite{pirandola2017fundamental}. Due to their shared functionality, quantum nodes and repeaters are considered network nodes.

\emph{3) Quantum links:} In quantum networks, a quantum link that connects two neighboring quantum nodes serves as the conduit for the distribution of EPR pairs, essentially functioning as a quantum channel. Quantum links exist in two main forms: optical fiber and free space. Both variants of quantum links inherently suffer from losses and decoherence \cite{orlowski2010sndlib}, resulting in the success probability of entanglement distribution between adjacent quantum nodes exponentially diminishing with the physical length of the quantum link \cite{van2013designing}. Consequently, two neighboring quantum nodes are compelled to undertake multiple entanglement distribution attempts to establish entanglement over a quantum link. Additionally, quantum memory can assist a quantum link in retaining multiple parallel EPR pairs.

\emph{3) Quantum memory:} Over the past few decades, there has been extensive exploration into various storage schemes for quantum memory \cite{lvovsky2009optical}. Advances in quantum memory technologies have led to practical improvements in coherence time, fidelity, and efficiency. Due to the low success probability of entanglement distribution between adjacent quantum nodes and the impressive $99.5$\% storage fidelity of quantum memory \cite{liu2020demand}, we have embraced a continuous model. This model entails the sharing of EPR pairs between adjacent quantum nodes before path selection, effectively addressing entanglement routing requests. Furthermore, the design of quantum memory allows for the creation of a composite of several independent and accessible memory units. This feature enables the assignment of a unique identity to each stored entanglement in memory, presented in the form of shared EPR pairs. This distinctive identification ensures the correct establishment of end-to-end entanglement between each S-D pair.

Managing the quantum network entails defining the quantum processor, channel, and repeaters. Classical networks connect quantum nodes, each equipped with computational and storage capabilities. The network operates synchronously, with activities segmented into discrete time slots \cite{li2022fidelity}. At the network layer, a centralized controller oversees network management. Nodes can report and update the data, such as fidelity and trace-distance, within 
the controller's records \cite{vardoyan2019stochastic}. The entanglement routing process includes: (i) Paired nodes generate entangled pairs, and the controller compiles these routing requests \cite{gyongyosi2019adaptive} and updates to each node's routing table. (ii) The controller then identifies routing paths and allocates resources utilizing a routing algorithm. Owing to resource constraints and multiple requests, some requests may be declined. (iii) Following the controller's directives, nodes carry out purification and swapping actions to establish multi-hop entanglement connections for the requested pairs.

\subsection{Quantum Trace-Distance Measure and Operations}
In quantum networks, to preserve information and establish E2E entanglement, we need to consider four unique operations: static, i.e., trace-distance and measurement of the quantum state, entanglement purification, and entanglement swapping. These operations have no direct equivalent in classical networking.

\emph{1) Static measurement of quantum state:} The static measure of a quantum state, also known as trace-distance, quantifies the distinguishability between two quantum states. It measures quantum information between quantum states and holds significant importance in quantum information theory \cite{xiang2007entanglement}. A smaller value of trace-distance indicates similarity between quantum states. In comparison, a larger value of trace-distance indicates greater distinguishability, which means the quantum state received at a destination node is different from the source node's quantum state. In a quantum network, the closeness of quantum states between adjacent nodes is quantified by the trace-distance \cite{nielsen2002quantum}, defined as:
\begin{equation}\label{eq1}
D_{u,v}\left(\rho, \sigma\right) \equiv \frac{1}{2} \operatorname{\textit{tr}}\left|\rho-\sigma\right|.
\end{equation} 
The trace-distance between two quantum states is denoted as $D_{u, v}$, where $\rho$ and $\sigma$ are the density matrices representing the quantum states at nodes $u$ and $v$ in a quantum network respectively. Where $\rho=\sum_{i}{p_i\left.\left|\psi_i\right.\right\rangle\left\langle\left.\psi_i\right|\right.}$ and  $\sigma=\sum_{i}{q_i\left.\left|\varphi_i\right.\right\rangle\left\langle\left.\varphi_i\right|\right.}$.
In the density matrix, $\left|\psi\right\rangle$ and $\left|\varphi\right\rangle$ represent the quantum state with a probability of $p_i$ and $q_i$, respectively. 
$tr$ is the trace operation, which is the sum of the diagonal elements of the matrix.
$\left|\rho-\sigma\right|$ refers to the positive semidefinite matrix, which is the absolute value of the difference between $\rho$ and $\sigma$. It is found by taking the matrix difference $\rho-\sigma$, and then computing its eigenvalues.
The factor of $\frac{1}{2}$ (\textit{see} Appendix~\ref{appendix1}) ensures that the trace-distance between the normalized density matrices takes values in the range [0,1], i.e., $0\le D_{u,v}\le1$ \cite{breuer2009measure}. This occurs when the two states $\left|\psi\right\rangle$ and $\left|\varphi\right\rangle$ are orthogonal, meaning they are completely distinguishable. Density matrices allow the description of mixed states, where the system is in a statistical mixture of noisy quantum states.

\emph{2) Closeness centrality:}
Closeness centrality is a measure used in network analysis to determine the importance of a node based on its proximity to all other nodes in the network. In a quantum network, a node with high closeness centrality using trace-distance can interact with or distribute entanglement to other nodes with minimal quantum state degradation. This centrality facilitates faster entanglement distribution. Routing algorithms can identify paths that minimize overall distance or communication cost by focusing on nodes with higher closeness centrality, leading to more efficient network performance. Additionally, closeness centrality allows for more effective resource allocation in quantum networks with limited resources such as qubits and entangled pairs. Nodes with higher closeness centrality are better positioned to maintain high-fidelity connections, ensuring optimal network operation.

\emph{3) Entanglement swapping:} 
To establish long-distance entanglement connections in a distributed quantum network, entanglement swapping is commonly used \cite{van2008system}. It converts one-hop entanglements into direct entanglements between non-adjacent nodes using a quantum repeater, as shown in Figure \ref{fig.fig1}. However, imperfect measurements during swapping can introduce noise and degrade the entanglement quality \cite{yen1971finding}. Furthermore, the varying fidelity of entangled pairs across different quantum channels requires different routing paths for end-to-end connections after entanglement swapping.

{\emph{4) Entanglement purification:} In practical quantum networks, directly establishing high-fidelity entanglement between adjacent nodes is challenging due to noise. This process improves link quality in quantum networks by employing controlled nodes and heralding stations \cite{dur2003entanglement,li2024survey}. Low-fidelity entangled pairs, such as EPR pairs, are purified using entanglement purification, typically involving CNOT gates or beam splitters \cite{dahlberg2019link}. The goal is to maximize the fidelity of entangled states by purifying EPR pairs. The fidelity of the resulting purified state is described in Lemma \ref{lemma1}.

\begin{lemma}\label{lemma1}
If in a quantum network, as shown in Figure \ref{fig.fig1}, each node has a density matrix, denoted as $\rho_i$ for quantum state $\left|\psi\right\rangle$ and $\sigma_i$ for quantum state $\left|\varphi\right\rangle$,
then the fidelity between $\rho_i$ and $\sigma_i$ is defined as the maximum value of the overlap $\left|\left\langle\psi\middle|\varphi\right\rangle\right|$, taken over all possible purification $\left|\psi\right\rangle$ of $\rho_i$ and $\left|\varphi\right\rangle$ of $\sigma_i$ can be expressed as:
\begin{equation}\label{eq12}
F\left(\rho_i,\sigma_i\right)=\max_{\left.|\psi\right\rangle,\left.|\varphi\right\rangle}{\left.|\left\langle\psi|\right.\varphi\right\rangle|}.
\end{equation}
\end{lemma}
\begin{proof}
For orthonormal bases, the source and destination states are represented as $\left|\psi_A\right\rangle$ on the source side and $\left|\varphi_B\right\rangle$ on the destination side. The source and destination are on the same channel, and the index may be assumed to run over the same set of values. Let's define an arbitrary quantum state 
 $\left|m\right\rangle$ = $\sum_{i}\left|\psi_A\right\rangle \left|\varphi_B\right\rangle$ and state  $\left|\psi\right\rangle$ be a purification of $\rho$. 
By Schmidt decomposition \cite{nielsen2002quantum}, we write the expression as, $\left|\psi\right\rangle$ = $\left(U_A\bigotimes\sqrt\rho U_B \right)$ $\left|m\right\rangle$ for some unitary operators $U_A$ and $U_B$ on system. 
Similarly, if state $\left|\varphi\right\rangle$ is any purification of $\sigma$ then there exist unitary operators $V_A$ and $V_B$ such that $\left|\varphi\right\rangle$ = $\left(V_A\bigotimes\sqrt\sigma V_B\right)$$\left|m\right\rangle$ so, taking the inner product of both state's values, we get:
\begin{equation}\label{eq12}
\left|\left\langle\psi\middle|\varphi\right\rangle\right|=\left|\left\langle\left.m\right|\right.\left(U_A^+V_A\bigotimes\ U_B^+\sqrt\rho\sqrt\sigma V_B\right)\left|\left.m\right\rangle\right.\right|.
\end{equation}
By applying the Hilbert-Schmidt inner product \cite{nielsen2002quantum} and entanglement on (\ref{eq12}), we get:
\begin{equation*}
    tr(A^+B)= \left|\left\langle m\right.\right|\left(A\bigotimes\ B\right)\left|\left.m\right\rangle\right|,
\end{equation*}
\begin{equation*}
    \left|\left\langle\psi\middle|\varphi\right\rangle\right|=\ \left|tr\left(U_AV_A^+\bigotimes\ U_B^+\sqrt\rho\sqrt\sigma V_B\right)\right|,                                              
\end{equation*}
\begin{equation*}
   \left|\left\langle\psi\middle|\varphi\right\rangle\right| =\left|tr\left({U_AV_A^+U_B^+\sqrt\rho\sqrt\sigma}V_B\right)\right|,
\end{equation*}
where $A^+$= $U_A$$V_A^+$, $B$ = $U_B^+\sqrt\rho\sqrt\sigma V_B$, and $U$= $V_B$$V_A^+$$U_A$$U_B^+$. Now we can rewrite the expression (\ref{eq12}) as:
\begin{equation} \label{eq13}
   \left|\left\langle\psi\middle|\varphi\right\rangle\right| = \left|tr\left(\sqrt\rho\sqrt\sigma U\right)\right|. 
\end{equation}
By using Cauchy-Schwarz inequality \cite{nielsen2002quantum}, for the Hilbert-Schmidt inner product on (\ref{eq13}), we get:
\begin{equation*}    \left|\left\langle\psi\middle|\varphi\right\rangle\right|\le\left|tr\left(\sqrt\rho\sqrt\sigma\right)\right|.
\end{equation*}
Hence, the fidelity between the states $\rho_i$ and $\sigma_i$ is:
\begin{equation} \label{eq-fidelity}
F\left(\rho_i,\sigma_i\right)=tr\left(\sqrt{\rho_i}\sqrt{\sigma_i}\right).
\end{equation}
\end{proof}

There are various entanglement purification techniques,  including symmetric purification \cite{dur1999quantum}, banded purification \cite{bassoli2021quantum}, and entanglement pumping \cite{van2008system}, among others. In this paper, we are using the entanglement pumping technique to purify EPR pairs. In this process, initially, the quantum nodes generate low-fidelity entangled EPR pairs. These pairs are probabilistically combined through entanglement purification using \eqref{eq-fidelity}, producing a higher-fidelity pair. The fidelity of an EPR pair is increased by purifying it with base-level pairs created via the physical entanglement mechanism. This iterative process continues until the desired higher fidelity of the EPR pairs is achieved. An example of the entanglement pumping method is shown in Figure \ref{pump}. Over multiple rounds, this pumping technique incrementally improves the entanglement quality, allowing the fidelity between quantum states $\left|\psi\right\rangle$ and $\left|\varphi\right\rangle$ to approach the maximum value. The pumping process continues until the desired fidelity threshold\footnote{The target fidelity level set for entangled quantum states, particularly EPR pairs, during quantum communication or quantum computing processes. In this paper, we use a fidelity threshold of $0.80$.} is achieved, at which point it stops. This recursive approach ensures that, despite noise, high-fidelity entanglement can be reliably established across quantum networks, enabling robust quantum communication.

\begin{figure}[!t]
\centering
    \includegraphics[width=0.40\textwidth]{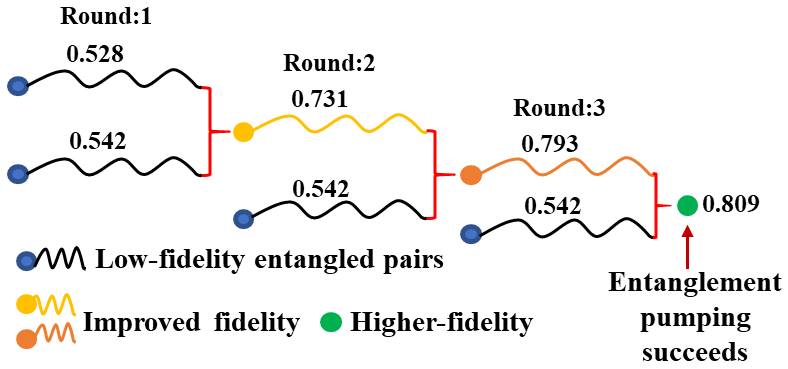}
    \caption{Entanglement purification using pumping process. In this process, quantum nodes initially generate low-fidelity EPR pairs, such as $\rho_i$ = $0.528$ and $\sigma_i$ = $0.548$. Once these pairs are generated, we apply the entanglement pumping process, using \eqref{eq-fidelity}, to purify the fidelity to $0.731$ in the first round. In the second round, we combine the newly purified fidelity with the base fidelity of $0.548$ to further purify the fidelity to $0.793$. This process is repeated in the next round, the fidelity is purified to $0.809$ which is greater than the desired fidelity threshold.}
    \label{pump}   
\end{figure}

\emph{6) Dynamic measurement of quantum state:} The dynamic measurement of a quantum state involves measuring a quantum system over time while considering the preservation of information. The environmental loss, imperfect swapping, and quantum gate errors cause a loss in fidelity. This is crucial for maintaining coherence and enabling accurate information processing in quantum systems.
The fidelity of the quantum channels is calculated based on the initial quantum state $\left|\psi\right\rangle$ of the nodes.
Quantum unitary operations are applied to protect the quantum state during communication \cite{xiang2007entanglement}.
The cumulative effect of these losses on the density matrix of the quantum state and fidelity of these states transferred over a quantum channel
$\varepsilon$ as:
%
\begin{equation}\label{eq3.1}
\rho_\varepsilon =\rho_{depol}\left(\rho_{dephase}\left(\rho_{fiber}\left(\rho\right)\right)\right),
\end{equation}
where $\rho_{\varepsilon}$ it represents the density matrix of quantum channel.
The quantum unitary operation minimizes fidelity loss by evaluating all possible initial states $\left|\psi\right\rangle$, thereby aiding information preservation.
This minimization is denoted as:
\begin{equation}\label{eq3}
F_{min}\left(\varepsilon\right)\equiv \min
[F\left(\left|\left.\psi\right\rangle\right.,\zeta\left(\left|\left.\psi\right\rangle\left\langle\left.\psi\right|\right.\right.\right)\right)].
\end{equation}
Where $F_{min}\left(\varepsilon\right)$ is the minimum fidelity of the quantum channel and it is designed to mitigate the loss in a quantum channel, ensuring information preservation and E2E entanglement fidelity within a quantum network.

\section{Quantum Network Routing Design} \label{QND}
The section introduces the trace-distance based path purification (TDPP) algorithm, which aims to establish E2E entanglement connections while preserving information in quantum networks. It identifies routing solutions for multiple S-D pairs in network $G$. TDPP selects entanglement paths with at most one purification link and maintains a trace-distance value of $0.60$ or higher. The goal is to safeguard information and maintain E2E entanglement fidelity in quantum networks \cite{xiang2007entanglement}.

\subsection{Information Preservation Quantification in Quantum Links}\label{IP}
We focused on preserving quantum state information and E2E entanglement fidelity between selected S-D pairs by performing dynamic operations on the quantum links using (\ref{eq1}) and (\ref{eq3}). The quantum operation is applied in real-time along the selected routing path $s\rightarrow r_2\rightarrow r_3\rightarrow d$, as shown in Figure \ref{fig:my_label_3}(c), when a quantum state with density matrix $\rho$ or $\sigma$ interacts with the adjacent nodes. Environmental noise during quantum communication can cause significant information loss through the quantum link. 
To address this problem, we perform a quantitative analysis to preserve information and entanglement fidelity.
We use Theorem \ref{theorem1} for our analysis showing that dynamically measuring the quantum state can effectively preserve information during communication in the quantum network.

While transmitting the state $\left|\psi\right\rangle$ from a source to a destination via a quantum channel, the discrete actions of the channel were governed by quantum operations, such as quantum state measurement or unitary operations on the quantum state, as described in \eqref{eq3}.
It is acknowledged that no quantum channel is completely flawless.
In practice, it is possible to measure information in real-time over a quantum communication channel even without knowing the specific details of the quantum state $\left|\psi\right\rangle$.
The behavior of the quantum channel could be quantified by E2E fidelity, illustrating that information is maintained during quantum communication. We analyzed an information preservation process on selected S-D pairs in Theorem \ref{theorem1} as follows:

\begin{theorem}
\label{theorem1}
In the TDPP algorithm, the selected path has an index set $X_i=$\{$p_i$, $q_i$\}, where $p_i$ and $q_i$ are the probability distribution of the quantum state at the source and destination node, associated with density matrix of the quantum state $\rho_i$ and $\sigma_i$, respectively, so the fidelity of the quantum channel can be expressed as:
\begin{equation}\label{eq14}
F\left(\sum_{i}{p_iq_i\ ,\sum_{i}{\rho_i\sigma_i}}\right)\geq\sum_{i}\sqrt{p_iq_i}F\left(\rho_i,\sigma_i\right).
\end{equation}
\end{theorem}

\begin{proof} 
The fidelity of a quantum channel is a measure of how well a quantum channel transmits quantum states. It is defined as the overlap between the output state of the channel and the corresponding input state. Mathematically, the fidelity of a quantum channel for the input state $\rho$ and the output state $\sigma$ is defined as:
\begin{equation}\label{eq14a}
F\left(\rho,\varepsilon\left(\rho\right)\right)=tr\left[\sqrt{\sqrt{\varepsilon\left(\rho\right)}\rho\sqrt{\varepsilon\left(\rho\right)}}\right]^2,
\end{equation}
where ${tr}$ denotes the trace operator. 
We are considering a quantum channel in which the quantum states at the source and destination nodes are represented by $\rho_i$ and $\sigma_i$, respectively, with associated probabilities $p_i$ and $q_i$ for each path. The total fidelity across all paths can be expressed using (\ref{eq14a}) as:
\begin{equation}
\begin{aligned}\label{eq14b}
\resizebox{0.45\textwidth}{!}{$
   F\left(\sum_{i}{p_iq_i},\sum_{i}{\rho_i\sigma_i}\right) =F\left(\sum_{i}{\sqrt{p_i}\sqrt{q_i}}, \sum_{i}{\frac{\sqrt{p_i}\rho_i\sqrt{q_i}\sigma_i}{\sqrt{p_i\rho_i}\sqrt{q_i\sigma_i}}} \right)$}.
\end{aligned}
\end{equation}
In (\ref{eq14b}), the left-hand side represents the fidelity of a quantum channel,
while the right-hand side describes the fidelity of the quantum channel for mixed quantum states, expressed as the sum of the individual fidelities weighted by the square roots of their respective probabilities.
The Jensen's inequality \cite{jensen1906fonctions} 
states that if $f$ is a concave function and $p_i$ are non-negative weights such that $\sum_{i=1}^n p_i = 1$, then:
\begin{equation} \label{jen}
    f\left(\sum_{i=1}^{n}{p_ix_i}\right)\geq\sum_{i=1}^{n}{p_if\left(x_i\right)}.
\end{equation} 
By using (\ref{jen}), we can rewrite (\ref{eq14b}) and 
express the fidelity across the entire quantum channel as:

\begin{equation*} \label{fidelity}
    F\left(\sum_{i}{p_iq_i},\sum_{i}{\rho_i\sigma_i}\right) \geq \sum_{i}{\sqrt{p_iq_i}F\left(\frac{\sqrt{p_i}\rho_i}{\sqrt{p_iq_i}},\frac{\sqrt{q_i}\sigma_i}{\sqrt{q_i\sigma_i}}\right)}.
\end{equation*}
Now, by applying the concavity of the fidelity \cite{miszczak2008sub}, where the states $\rho_i$ and $\sigma_i$ are associated with the source and destination probability distributions $p_i$ and $q_i$, we obtain:
  \begin{equation*} 
  F\left(\sum_{i}{p_iq_i\ ,\sum_{i}{\rho_i\sigma_i}}\right)\geq\sum_{i}\sqrt{p_iq_i}F\left(\rho_i,\sigma_i\right).
  \end{equation*}
%
The concavity ensures that the overall fidelity is at least as large as the weighted sum of the individual fidelities. 
\end{proof}
The (\ref{eq14}) demonstrates the preservation of information and fidelity within the quantum network, ensuring that the quantum state remains intact following purification on a quantum link. It outlines the conservation of fidelity and information across quantum networks, spanning from E2E within the network. 
In the worst-case scenario, if the source node attempts to minimize the initial state $\left.|\psi\right\rangle$, the resulting fidelity of the link after this quantum state minimization process will be:
\begin{equation}\label{eq16}
F_{min}\left(\varepsilon\right)\equiv\min_{\left.|\psi\right\rangle}{F}\left(\left.|\psi\right\rangle,\zeta\left(\left.|\psi\right\rangle\left\langle\psi|\right.\right)\right),
\end{equation}
where $F_{min}\left(\varepsilon\right)$ denotes minimum fidelity of quantum channel. 
As an illustration, consider a quantum network employing a polar channel for entanglement generation across all quantum states $\left.|\psi\right\rangle$. In such a scenario, the minimized fidelity can be expressed as $F_{min}\left(\sqrt{1-\frac{p}{2}}\right)$, with $p$ representing the probability of the quantum state $\left.|\psi\right\rangle$. 
Now, in the context of a phase-damping channel under similar conditions, the fidelity is given by (\ref{neq}), as follows:
\begin{equation} \label{neq}   F\left(\left.|\psi\right\rangle,\zeta\left.|\psi\right\rangle\left\langle\psi|\right.\right) = \sqrt{p + (1-p)\langle\psi|Z|\psi\rangle^2},
\end{equation} 
where $p$ denotes a probability, and $Z$ represents a specific operator acting on the quantum state $\left.|\psi\right\rangle$.
For a pure quantum state, the right-hand side of the fidelity square term in (\ref{neq}) should always be non-negative. In light of this, we choose the minimum fidelity $F_{min}\left(\varepsilon\right)=\sqrt p$ across the quantum link. This minimization strategy aims to protect the information from potential channel noise, ensuring the overall preservation of the pure quantum state in the network using (\ref{eq16}).
Accordingly, 
we formulated the information preservation problem to maximize the fidelity of the quantum state as follows:

\begin{equation}
\text{Maximize:} \quad \hat{F}_{u,v}\left(\rho,\varepsilon\left(\rho\right)\right),  \\
\label{obj}
\end{equation}

\begin{subequations}
\text{subject to:}
\begin{equation}
\ c_{u,v} \geq 1, \ \forall \left(u,v\right) \in E,  \\
\label{obj1}
\end{equation}
\begin{equation}
 m_u \geq 1, \ \forall u \in \mathcal{V},  \\ 
\label{obj2}
\end{equation}
\begin{equation}
\sum_{u:\left(u,v\right) \in E}{F_{u,v}\left(\rho,\sigma\right)-\sum_{v:\left(v,u\right) \in E}{F_{v,u}\left(\sigma,\rho\right) = p_i,}} \ \forall u=s_k, \\
\label{obj3}
\end{equation}
\begin{equation}
\sum_{u:\left(u,v\right) \in E}{F_{u,v}\left(\rho,\sigma\right)- \sum_{v:\left(v,u\right) \in E}{F_{v,u}\left(\sigma,\rho\right) = q_i,}} \ \forall u=d_k, \\
\label{obj4}
\end{equation}
%
\begin{flalign}
\sum_{u:\left(u,v\right)\in E}{F_{u,v}\left(\rho,\sigma\right)- \sum_{v:\left(v,u\right) \in E}{F_{v,u}\left(\sigma,\rho\right)=0,}} \nonumber \\ \forall v \in  \mathcal{V}-\{s_k, d_k\}, 
\label{obj5}
\end{flalign}
%
\begin{equation}
F_{u,v}\left(\rho,\sigma\right)+F_{u,v}\left(\sigma,\rho\right)\ \le\ c_{u,v},\ \ \forall\left(u,v\right)\ \in\ E,
\label{obj6}
\end{equation} 
\begin{equation}
\sum_{i}{\sqrt{p_i,q_i}F_{u,v}}\left(\rho,\sigma\right) \le F_{u,v}\left(\sum_{i}{p_i,q_i}\right), \ \forall \left(u,v\right) \in E.  
\label{obj7}
\end{equation}
\end{subequations}

The objective of (\ref{obj}) is to maximize the fidelity of quantum state to preserve the quantum information and E2E entanglement fidelity, i.e., the E2E entanglement connection established between S-D pairs. 
The first two constraints, i.e., (\ref{obj1})-(\ref{obj2}), show that quantum nodes have a finite number of quantum channels and quantum memory, respectively. The next three constraints, i.e., (\ref{obj3})-(\ref{obj5}) are the flow conservation constraints for the source node, destination node, and any intermediate nodes, respectively, that are held in all routing problems. 
The inequality in (\ref{obj6}) shows that the number of entanglement link fidelity generated for E2E entanglement paths should not exceed the channel capacity.
In (\ref{obj7}) states that after performing purification over on selected path, a quantum state's fidelity across a quantum edge must be greater than or equal to the previous fidelity value.

\subsection{Trace-Distance based Path Purification Algorithm} \label{TDPP}
In this section, we introduce the TDPP algorithm, which identifies potential paths for E2E entanglement and conducts purification operations on quantum links when necessary.
The algorithm, presented in Algorithm~\ref{alg:algorithm1}, systematically explores all potential paths based on the closeness centrality of nodes for each S-D pair in the networks. It updates the path's trace-distance and fidelity using the routing cost of closeness centrality as shown in Figure \ref{fig:my_label_3}(a). During each iteration, the algorithm compares the trace-distance of the chosen paths with the fidelity of each edge. When the trace-distance of a selected path exceeds the edge's fidelity, the algorithm initiates purification to maximize the fidelity and maintain the entanglement connection over that edge.

\begin{figure*}[!t]
\centering
    \subfloat[Core network with trace-distance and fidelity with closeness centrality at nodes]{
    \includegraphics[width=0.3\textwidth]{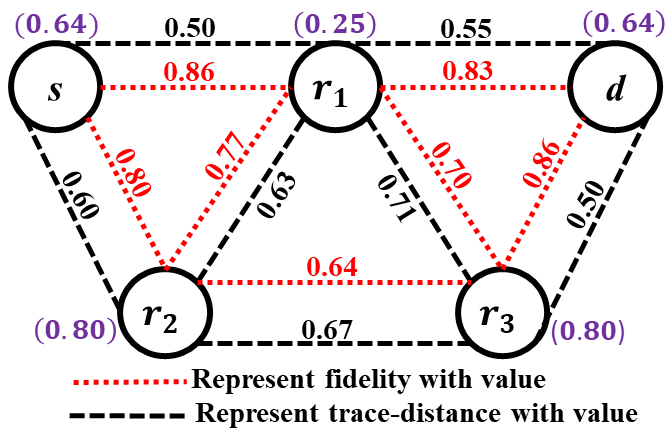}}
    \qquad
    \subfloat[Potential path based on closeness centrality]{
    \includegraphics[width=0.3\textwidth]{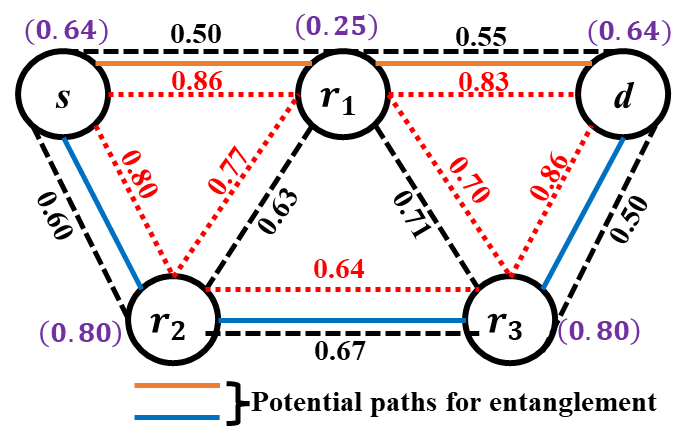}}  
    \qquad
    \subfloat[Selection of path for communication]{
    \includegraphics[width=0.3\textwidth]{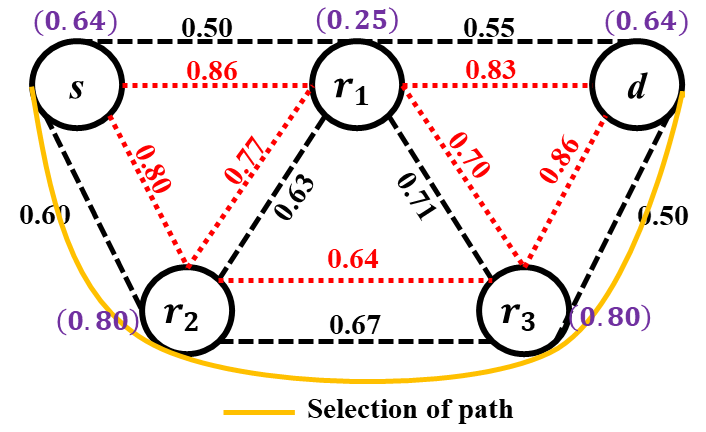}}
    \qquad
    \subfloat[Identify condition for information preservation]{
    \includegraphics[width=0.3\textwidth]{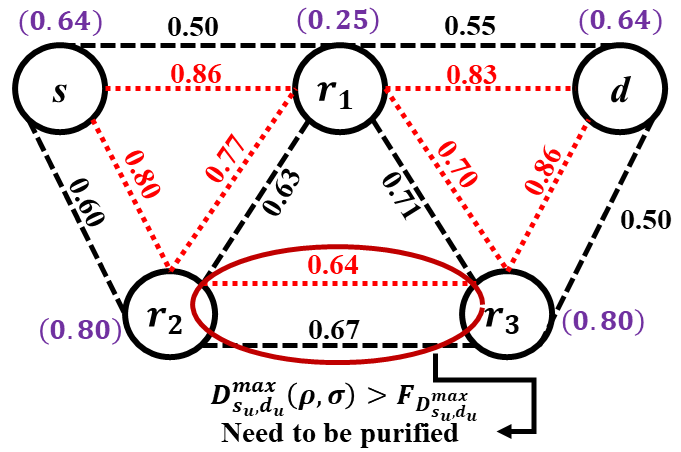}}
    \qquad
    \subfloat[Identify maximum fidelity threshold in selected path]{
    \includegraphics[width=0.3\textwidth]{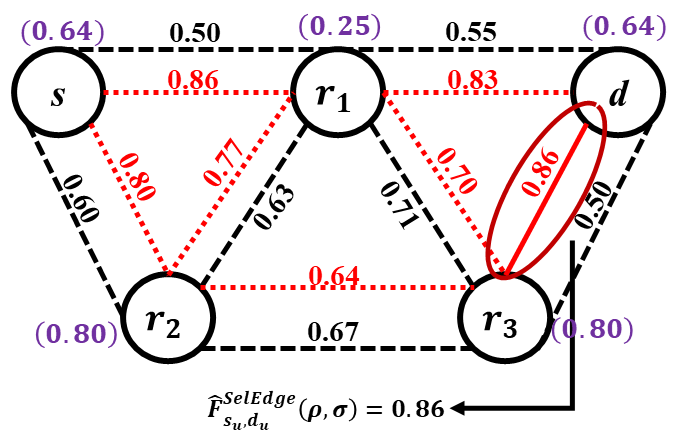}}
    \qquad
    \subfloat[Optimal path for information preservation and E2E fidelity entanglement]{
    \includegraphics[width=0.3\textwidth]{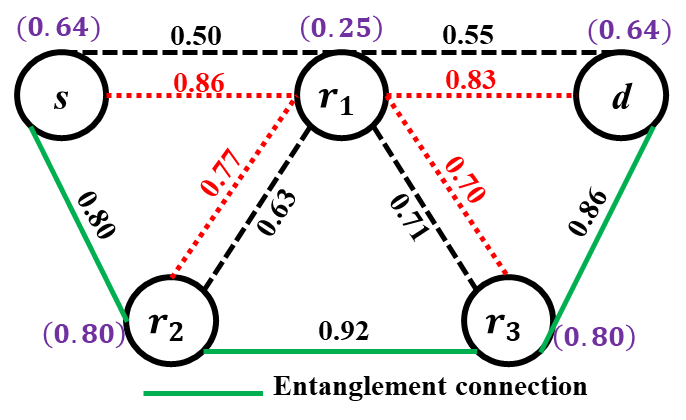}}
    \qquad
    \caption{Trace-distance based quantum network routing architecture, where fidelity values are indicated in red, trace-distance in black, and the closeness centrality of each node in purple.}
    \label{fig:my_label_3}   
\end{figure*}

\RestyleAlgo{ruled}
\begin{algorithm}[!t]
\caption{TDPP Algorithm}
\label{alg:algorithm1} 
\KwData{$G =[\mathcal{V}, E], F[\ ], D[\ ]$;}
\KwResult{Purification decision, $F_{D_{s_k,d_k}^{max}}^{purific}\left(\rho,\sigma\right)$;}
\SetKwRepeat{Do}{do}{while}
\textbf{Step 1. Initialization:}  \\
   Calculate closeness centrality $C_{v}$, $v$ $\in$ $\mathcal{V}$ using (\ref{eq5}); \\
   Calculate trace-distance $D_{u,v}$ $\in$ \textit{E} using (\ref{eq6}); \\
   Calculate fidelity $F_{u,v}$ $\in$ \textit{E} using (\ref{eq7}); \\
   Construct auxiliary graph $G^a = (\mathcal{V}, E^a)$, $E^a$ $\subseteq$ $E$; \\
   $Q \gets$ Priority queue according to closeness centrality value of $C_{v}$; \\
\textbf{Step 2. Path selection and purification process:}  \\
\For{$minCost= C_{v}^{max}  : \left|E\right|c_{u,v}$} 
{   
        Multiple shortest paths set $P_{minCost}^{SPF}$ $\gets$ using \textit{K}-shortest path algorithm \cite{yen1971finding}; \\
          \If{no path available for ($s_k$,$d_k$)}
            {
             break; \\
            } 
          \For{ $P_{i}(s_k,d_k) \in P_{minCost}^{SPF}$}
            {
               
                   \For{$u$ in range ($n$)}
                   {
                       \For{$v$ in range ($n$)}
                       {
                          \If{$D[u][v] < F[u][v]$}
                            {
                            no need for purification; \\
                            }
                        }
                    }
                 
    $D^{max}_{s_k,d_k}\left(\rho,\sigma\right)$ = update maximum trace-distance $D_{u,v}$ over an edge using (\ref{eq8}); \\
   $\hat{F}_{s_k,d_k}^{SelEdge}\left(\rho,\sigma\right)$ = update maximum fidelity $F_{u,v}$ over an edge using (\ref{eq9}); \\
            }  
    \Do{ \text{check condition using (\ref{eq18b})}}{$F_{D_{s_k,d_k}^{max}}^{purific}\left(\rho,\sigma\right)$ = purification using (\ref{eq18a});}
    \textbf{return} {$F_{D_{s_k,d_k}^{max}}^{purific}\left(\rho,\sigma\right)$;}
}  
\end{algorithm}

The TDPP algorithm prioritizes closeness centrality over the number of hops typically favored in classical routing algorithms, 
such as Dijkstra algorithm \cite{caleffi2017optimal, santos2023shortest}, to ensure the chosen routing path is the most reliable for information flow in quantum networks. The algorithm takes into account the factors that secure the preservation of information and entanglement fidelity between S-D pairs, ensuring the selected path is optimal for this purpose. This process comprises two specific steps as follows:

\emph{1) \textbf{Initialization:}}The first step in the TDPP procedure is to measure closeness centrality. 
It quantifies the average length of the shortest path from a node to all other nodes in the network. Formally, for a given node $v$, the closeness centrality $C_{v}$ \cite{bloch2023centrality}, measures to which extent a node $v$ is near to all the other nodes along the shortest paths, where $v \in \mathcal{V}$ in a network $G$:  

\begin{equation}\label{eq5}
C_{v}=\frac{N-1}{\sum_{u\neq v}d\left(v,u\right)}.
\end{equation}
Where $N$ is the total number of nodes in the network and $d\left(v,u\right)$ which is the shortest path length between nodes $v$ and $u$. This metric is particularly useful for identifying influential nodes in terms of information spread or communication efficiency within the network.
Consequently, closeness centrality helps to identify the potential paths for entanglement generation in the network.
 For example, in the given network \ref{fig:my_label_3}(b), there are two potential paths available $s\rightarrow r_2\rightarrow r_3\rightarrow d$, and $s\rightarrow r_1\rightarrow d$, receptively. 
However, given the closeness centrality values shown in Figure~\ref{fig:my_label_3}(b), the closeness is larger in the path $s\rightarrow r_2\rightarrow r_3\rightarrow d$ compared to the path $s\rightarrow r_1\rightarrow d$.
Hence, we select the $s\rightarrow r_2\rightarrow r_3\rightarrow d$ path for both communication and E2E entanglement distribution, as depicted in Figure \ref{fig:my_label_3}(c).
We devise an iterative method to determine the most efficient route at the minimum cost possible.
The method uses the cost as the basis for each iteration, which helps guide the path search. We use the $K$-shortest path algorithm \cite{yen1971finding} to identify several equally low-cost routes.
After identifying low-cost routes, we find out the trace-distance of quantum states of the selected routes $P_{i}(s_k,d_k)$ as shown in Figure \ref{fig:my_label_3}(c), i.e., $D_{u,v}\left(\rho,\sigma\right)$ for information preservation over the networks. The trace-distance $D_{u,v}\left(\rho,\sigma\right)$ can be calculated as:
 \begin{equation}\label{eq6}
D_{u,v}(\rho,\sigma)=\frac{1}{2}(\rho-\sigma),
\end{equation}
where $\left(u,v\right)\in E$ and $\left(\rho,\sigma\right)$ is the density matrix. 
Assign the trace-distance and fidelity values over an edge to establish an entanglement connection between the S-D pairs.
The selected route $P_{i}(s_k,d_k)$, i.e., $s\rightarrow r_2\rightarrow r_3\rightarrow d$ needs to check the trace-distance and fidelity values over the edge $\left(u,v\right)$. 
It demonstrates that the degradation of quantum states and the loss of information in the network are reflected by the maximum trace-distance value over the edge, in contrast to fidelity, as illustrated in Figure \ref{fig:my_label_3}(d). So, the fidelity of the selected path needs to be improved by using extra quantum operations on path $P_{i}\left(s_k,d_k\right)$. We can calculate it using Lemma \ref{lemma1} as:
%
 \begin{equation}\label{eq7}
 F_{u,v}\left(\rho,\sigma\right)=tr\sqrt{{\rho}^\frac{1}{2},{\sigma}^\frac{1}{2}}.
 \end{equation}
 %
 Where $F_{u,v}\left(\rho,\sigma\right)$ denotes the fidelity over the edge ($u, v$). For example, in Figure \ref{fig:my_label_3}(d) the selected minimum cost path for entanglement generation is $s\rightarrow r_2\rightarrow r_3\rightarrow  d$. 
 The trace-distance of the links $\left(s,r_2\right)$, $\left(r_2,r_3\right)$ and $\left(r_3,d\right)$ are 0.60, 0.67, and 0.50, and the corresponding fidelity of the links are 0.80, 0.64, and 0.86, respectively.

\emph{2) \textbf{Path Selection and Purification Process:}} 
The TDPP method determines the most efficient routing path by considering the minimal cost associated with a node in quantum networks. This method employs the node's minimum cost as a metric for each iteration, streamlining both path selection and search processes.
This approach uses the \textit{K}-shortest path algorithm \cite{yen1971finding} to determine multiple shortest paths with the same minimal closeness distance between quantum nodes. 

After the generation of entangled pairs, entanglement switching is necessary to connect the end-to-end entanglement connection. Considering the environmental noise and imperfect measurement of the quantum gate (Pauli-X-gate), the selected path faces where the maximum trace-distance and degradation of fidelity. 
The purification decision process is designed to ensure the preservation of end-to-end information and fidelity by minimizing the trace-distance. This involves checking the condition using \eqref{new}  step-by-step to achieve the desired outcome.
\begin{equation}\label{new}
    F_{u,v}\left(\rho,\sigma\right)\geq1-\frac{1}{2}D_{u,v}\left(\rho,\sigma\right).
\end{equation}
To provide all possible purification options and corresponding costs of selected S-D pairs trace-distance, and fidelity over an edge, we calculated the updated cost for edge $\left(u,v\right)\in\ E$ using (\ref{eq8}), (\ref{eq9}), and (\ref{eq:18}) as follows:
\begin{equation}\label{eq8}
\begin{split}
D_{s_k,d_k}^{max}(\rho,\sigma) =\max (D_{s_k,r_i}\left(\rho,\sigma\right), D_{r_i,s_k}\left(\sigma,\rho\right),  \cdots, &\\ D_{r_j,d_k}\left(\rho,\sigma\right),  D_{d_k,r_j}\left(\sigma,\rho\right)), \forall\ s_k, r_i, r_j, d_k \in \mathcal{V},
\end{split}
\end{equation}
\begin{equation}\label{eq9}
\begin{split}
{\hat{F}}_{s_k,d_k}^{SelEdge}(\rho,\sigma) = \max (F_{s_k,r_i}\left(\rho,\sigma\right),F_{r_i,s_k}\left(\sigma,\rho\right), & \\  \cdots, F_{r_j,d_k}\left(\rho,\sigma\right), F_{d_k,r_j}\left(\sigma,\rho\right)), \forall\ s_k, r_i, r_j, d_k \in \mathcal{V},
\end{split}
\end{equation}
\begin{subnumcases}{\label{eq:18}}
 \label{eq18a}   \!\!  F_{D_{s_k,d_k}^{max}}^{purific}\left(\rho,\sigma\right)=\sqrt{{\hat{F}}_{s_k,d_k}^{SelEdge}\left(\rho,\sigma\right)},\\
 \label{eq18b}   \!\!F_{D_{s_k,d_k}^{max}}^{purific}\left(\rho,\sigma\right)\geq{\hat{F}}_{s_k,d_k}^{SelEdge}\left(\rho,\sigma\right),
 \end{subnumcases}
where $D^{max}_{s_k,d_k}\left(\rho,\sigma\right)$ is the maximum trace-distance of an edge in the selected path, $\hat{F}_{s_k,d_k}^{SelEdge}\left(\rho,\sigma\right)$ is maximum fidelity of an edge in the selected path and $F_{D_{s_k,d_k}^{max}}^{purific}\left(\rho,\sigma\right)$ is purified fidelity of the edge $(u,v)$ after performing purification operation on a selected path in S-D pairs. The fidelity of edge $(u,v)$, i.e., $F_{s_k,d_k}\left(\rho,\sigma\right)$ is calculated during purification operation as in Lemma~\ref{lemma1} to achieve high-fidelity entanglement generation. 
For example, in the selected path $s\rightarrow r_2\rightarrow r_3\rightarrow d$ for entanglement generation, we select the maximum fidelity of the path as threshold fidelity, as shown in Figure \ref{fig:my_label_3}(e), and the maximum trace-distance of the selected path after an update is $\hat{F}_{s_k,d_k}^{SelEdge}\left(\rho,\sigma\right)$$=0.86$, and $D^{max}_{s_k,d_k}\left(\rho,\sigma\right)$$=0.67$, respectively. 
This means that along the selected path, the link with the maximum trace distance is precisely the one with the lowest fidelity.
Hence, we need to purify the link where the trace-distance is maximum to preserve the information in the quantum network.
For this, we use \eqref{eq18a} and \eqref{eq18b} to purify the link fidelity of selected paths  $s\rightarrow r_2\rightarrow r_3\rightarrow  d$.
Then the resulting fidelity after improvement will be $F_{D_{s_k,d_k}^{max}}^{purifc}\left(\rho,\sigma\right)$$=0.92$, as shown in Figure \ref{fig:my_label_3}(f). Thus, fidelity improvement is increasing along with decreasing the trace-distance during the purification. Finally, the optimal path for entanglement is shown in Figure \ref{fig:my_label_3}(f).

However, the purification constraint $D[u][v] < F[u][v]$ must be satisfied for the routing path. If an edge $(u, v)$ cannot provide entangled pairs that satisfy $D^{max}(s_k,d_k) \ge F_{D_{s_k,d_k}^{max}}^{purifc}\left(\rho,\sigma\right)$ even after purification, it must be removed from the graph $G$ to reduce complexity. Once this is done, TDPP creates an updated graph $G^a$ to keep track of the purification decisions. Finally, TDPP uses the closeness centrality (\ref{eq5}) to find the path in graph $G$, ensuring the minimum possible cost. Throughout the iteration process, TDPP also constructs a priority queue to store potential routing paths $P_{minCost}^{SPF}$ and repeats the process as needed.

\begin{figure*}[!t]
\centering
    \subfloat[$\alpha$ = 0.4, $\beta$ = 0.4]{
    \includegraphics[width=0.3\textwidth]{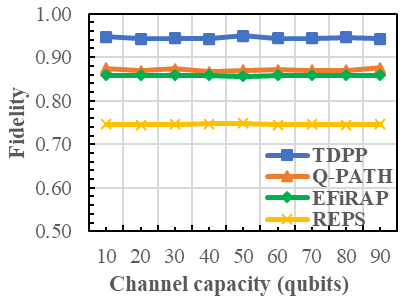}}
    \qquad
    \subfloat[$\alpha$ = 0.4, $\beta$ = 0.5]{
    \includegraphics[width=0.3\textwidth]{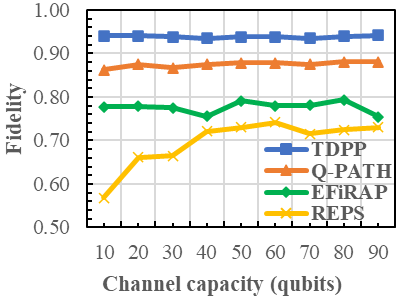}}
    \qquad
    \subfloat[$\alpha$ = 0.4, $\beta$ = 0.6]{
    \includegraphics[width=0.3\textwidth]{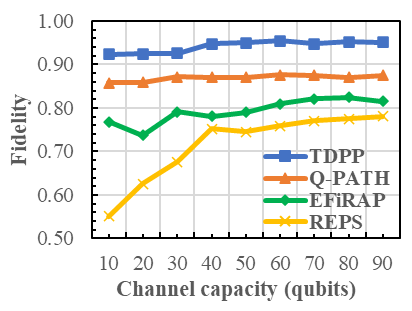}}
    \qquad
    \subfloat[$\alpha$ = 0.5, $\beta$ = 0.4]{
    \includegraphics[width=0.3\textwidth]{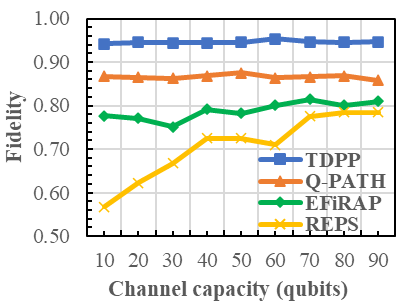}}
    \qquad
    \subfloat[$\alpha$ = 0.5, $\beta$ = 0.5]{
    \includegraphics[width=0.3\textwidth]{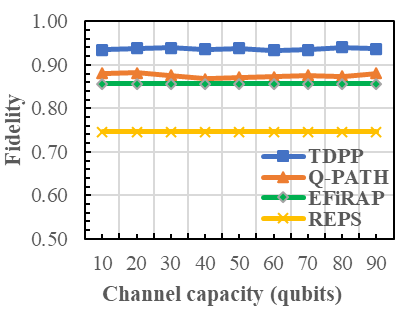}}
    \qquad
    \subfloat[$\alpha$ = 0.5, $\beta$ = 0.6]{
    \includegraphics[width=0.3\textwidth]{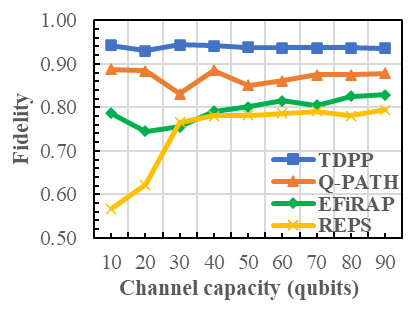}}
    \qquad
    \subfloat[$\alpha$ = 0.6, $\beta$ = 0.4]{
    \includegraphics[width=0.3\textwidth]{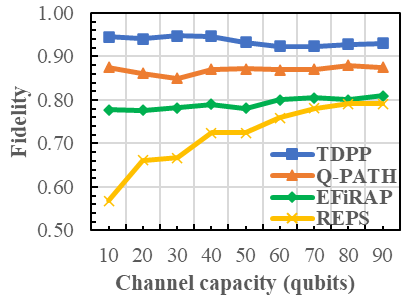}}
    \qquad
    \subfloat[$\alpha$ = 0.6, $\beta$ = 0.5]{
    \includegraphics[width=0.3\textwidth]{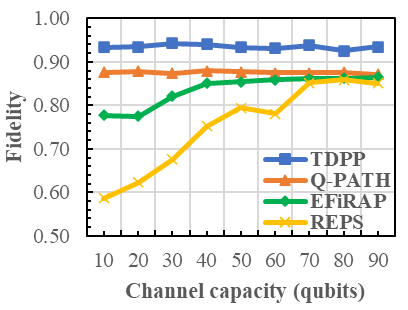}}
    \qquad
    \subfloat[$\alpha$ = 0.6, $\beta$ = 0.6]{
    \includegraphics[width=0.3\textwidth]{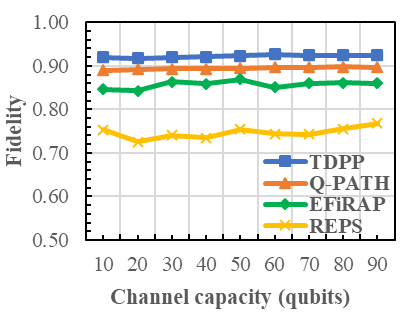}}
    \caption{Fidelity performance comparison of multiple S-D pairs in communication networks.}
    \label{fig:my_label_1}   
\end{figure*}

\section{Performance Analysis} \label{PA}
\subsection{Simulation Methodology} 
We conducted a series of numerical evaluations to implement the proposed routing algorithms in quantum networks. The simulation platform utilized for this purpose featured an AMD Ryzen 7 3700X 3.6GHz CPU with 32GB RAM and ran on a Windows-10 64-bit system. As quantum networks evolve, they are anticipated to become essential infrastructure for secure communications and quantum applications, akin to the current Internet backbone. Therefore, we utilized the US backbone network \cite{orlowski2010sndlib} as the topology for our simulation. We conducted simulations $1000$ times for each parameter configuration and calculated the average results. The initial fidelity of entangled pairs adhered to a normal distribution with a mean of $0.80$ and a standard deviation of $0.1$. The standard qubit lifetime is $1.46$ seconds, so we used a synchronization timestamp of 500ms. Each node is assigned $20$ qubits of quantum memory, and the channel capacity of the links between adjacent nodes ranges from $10$ to $90$ qubits.

As a baseline, we used the routing scheme proposed in \cite{li2021effective}, which includes a purification operation.
The entanglement generation rate, which represents the number of successfully established entangled pairs in the network during a fixed time window, is referred to as the throughput (qubits/slot) of the system in routing terminology.
This scheme also adopts a proportional share for resource allocation.
In our experiment, we utilized the US backbone network as the network topology. We employed quantum state $\left.\left|\psi\right.\right\rangle=\alpha\left.\left|0\right.\right\rangle+\beta\left.\left|1\right.\right\rangle$ where $\alpha$ and $\beta$ act as probability amplitude of getting state $0$ and state $1$, respectively.
For example, $\left.\left|\psi\right.\right\rangle=$0.4$\left.\left|0\right.\right\rangle+$0.5$\left.\left|1\right.\right\rangle$. 
These probability amplitude coefficients ($\alpha$ and $\beta$) need to be normalized to ensure that they fall within the range of [$0$, $1$]. 
Therefore, we select the maximal mixed quantum state within the range of (0.4, 0.5, 0.6), which indicates greater mixedness and lower information content. This choice impacts the purity of quantum states used in the network and the efficiency of quantum information processing.
To gauge the effectiveness of our algorithm, we compared it against three other existing algorithms, namely REPS \cite{zhao2021redundant}, EFiRAP \cite{zhao2022e2e}, and Q-PATH \cite{li2022fidelity}. This comparison was made regarding routing and purification techniques, network throughput, and fidelity.

\begin{figure*}[!t]
\centering
    \subfloat[$\alpha$ = 0.4, $\beta$ = 0.4]{
    \includegraphics[width=0.3\textwidth]{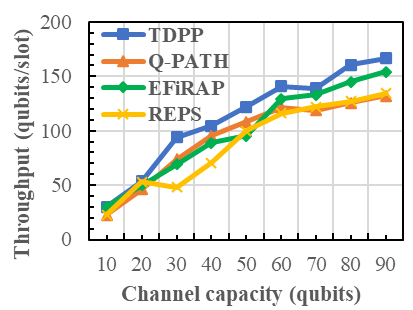}}
    \qquad
    \subfloat[$\alpha$ = 0.4, $\beta$ = 0.5]{
    \includegraphics[width=0.3\textwidth]{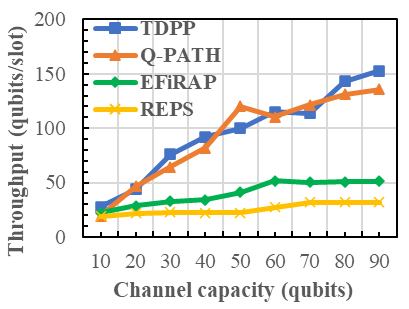}}
    \qquad
    \subfloat[$\alpha$ = 0.4, $\beta$ = 0.6]{
    \includegraphics[width=0.3\textwidth]{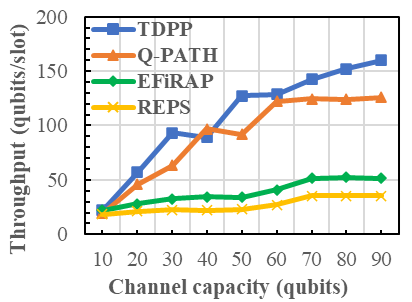}}
    \qquad
    \subfloat[$\alpha$ = 0.5, $\beta$ = 0.4]{
    \includegraphics[width=0.3\textwidth]{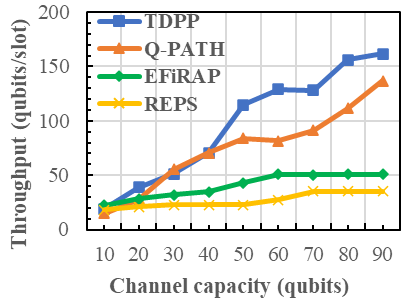}}
    \qquad
    \subfloat[$\alpha$ = 0.5, $\beta$ = 0.5]{
    \includegraphics[width=0.3\textwidth]{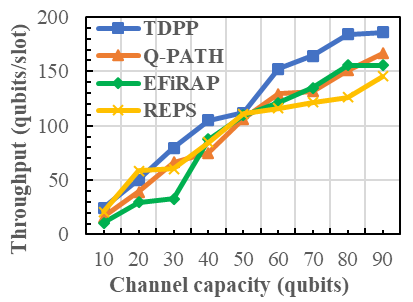}}
    \qquad
    \subfloat[$\alpha$ = 0.5, $\beta$ = 0.6]{
    \includegraphics[width=0.3\textwidth]{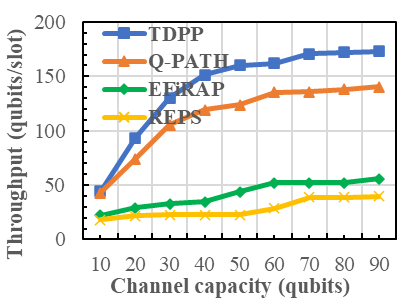}}
    \qquad
    \subfloat[$\alpha$ = 0.6, $\beta$ = 0.4]{
    \includegraphics[width=0.3\textwidth]{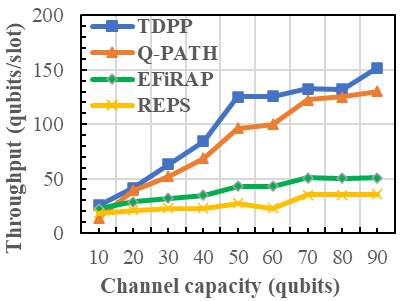}}
    \qquad
    \subfloat[$\alpha$ = 0.6, $\beta$ = 0.5]{
    \includegraphics[width=0.3\textwidth]{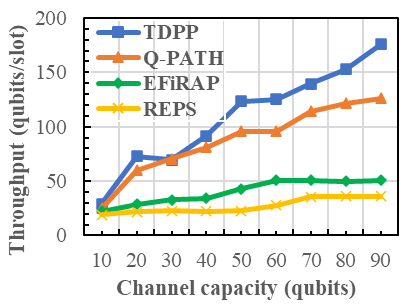}}
    \qquad
    \subfloat[$\alpha$ = 0.6, $\beta$ = 0.6]{
    \includegraphics[width=0.3\textwidth]{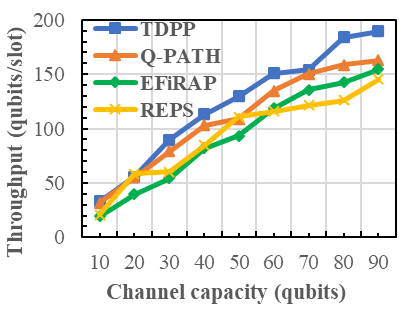}}
    \caption{Throughput-performance comparison of multiple S-D pairs in communication networks.}
    \label{fig:my_label_2}   
\end{figure*}

\subsection{Fidelity Performance Trade-off Analysis} 
Figures ~\ref{fig:my_label_1}(a)$-$(i) show the fidelity performance of the algorithms under different probability amplitude ranges of $\alpha$ and $\beta$ values of ($0.4$, $0.5$, $0.6$). In all scenarios of $\alpha$ and $\beta$, our proposed TDPP algorithm achieves high-fidelity entangled pairs. This outcome is due to our algorithm performing purification based on trace-distance rather than the fidelity threshold. Upon careful observation of these results, we can categorize them into three cases: \textbf{case 1.} where $\alpha = \beta$ depicted in Figures ~\ref{fig:my_label_1}(a), (e), and (i); \textbf{case 2.} $\alpha < \beta$ as in Figures~\ref{fig:my_label_1}(b), (c), and (f); and \textbf{case 3.} $\alpha > \beta$ as in Figure~\ref{fig:my_label_1}(d), (g), and (h).

\emph{1) \textbf{$\alpha = \beta$:}} When we set the $\alpha =  \beta$ = 0.4, 0.5, and 0.6 in Figures \ref{fig:my_label_1}(a), (e), and (i), respectively, our proposed TDPP algorithm performs better compared to other existing algorithms in all three scenarios. 
This is because in TDPP, the difference in trace-distance between the nodes becomes minimal, and the fidelity of the link is high.
In Figures \ref{fig:my_label_1}(a) and (e), the fidelity difference between the Q-PATH and EFiRAP algorithms is very minimal. This is because, when the channel capacity is high, multi-round purification operations are needed to meet the end-to-end fidelity requirement. As a result, the number of available entangled pairs on each quantum channel decreases. If all entangled pairs reach their maximum, all the routes are discovered and entangled in solution time-space, leading to a similar routing solution obtained by Q-PATH and EFiRAP algorithms. However, in Figure \ref{fig:my_label_1}(i), the performance of Q-PATH and EFiRAP is clear in terms of fidelity measurement with a utility factor of 0.6. In all three scenarios, the performance of REPS is the lowest, as it only considers entangled pairs for connection and does not require purification or threshold fidelity values to establish the connection.

\emph{2) \textbf{$\alpha < \beta$:}}
Figures~\ref{fig:my_label_1}(b), (c), and (f) show the scenario where $\alpha <\beta$. We observed that the TDPP algorithm performed better than all other algorithms, like in case-1, i.e., $\alpha = \beta$. However, the fidelity performances are not stable as in case-1. The performance of Q-PATH is relatively better than the other two algorithms. In low channel capacity, EFiRAP performs better than REPS algorithm. With the increase in channel capacity, this gap gradually decreases. This is because in the EFiRAP algorithm, entanglement is generated based on the fidelity threshold, and the REPS algorithm generates an entangled path based on available quantum resources. So, when channel capacity increases, generating a more entangled path in quantum networks is possible, showing the REPS algorithm shows linearly increasing fidelity values.

\emph{3) \textbf{$\alpha > \beta$:}}
Figures~\ref{fig:my_label_1}(d), (g), and (h) depict scenarios where $\alpha >\beta$. Our observations indicate that the performance of the algorithms is quite similar to the case-2 scenario, i.e., $\alpha <\beta$. However, the performance of the TDPP algorithm still surpasses all other algorithms. The performance of the algorithms is somewhat unstable at low channel capacities, but it stabilizes as the capacity increases.

\subsection{Throughput Performance Analysis} 
Figures~\ref{fig:my_label_2}(a)$-$(i) show the throughput performance of the algorithms in different ranges of utility factor $\alpha$ and $\beta$ values. As in Figure~\ref{fig:my_label_1}, we can group the results of Figure~\ref{fig:my_label_2} into three categories as \textbf{case 1.} $\alpha = \beta$ as in Figures~\ref{fig:my_label_2}(a), (e), and (i), \textbf{case 2.} $\alpha < \beta$ as in Figures~\ref{fig:my_label_2}(b), (c), and (f), and \textbf{case 3.} $\alpha > \beta$ as in Figures~\ref{fig:my_label_2}(d), (g), and (h).  

\emph{1) \textbf{$\alpha = \beta$:}} Figures \ref{fig:my_label_2}(a), (e), and  (i) show the TDPP algorithm has better throughput performance as compared to other existing algorithms.
This is because the TDPP algorithm works on the closeness centrality, which means that the quantum nodes are close to all other nodes in the network, and it helps to generate high-fidelity qubits. The throughput in other algorithms is also relatively stable in this scenario, i.e., $\alpha = \beta$. With the increased channel capacity, the throughput increases steadily in all cases. Because the solution space of a purifying operation is related to channel capacity, algorithms with higher channel capacity can identify superior solutions.

\emph{2) \textbf{$\alpha < \beta$:}} 
When we set the utility factor $\alpha < \beta$ in Figures \ref{fig:my_label_2}(b), (c), and (f), the TDPP algorithm has better throughput performance than other algorithms. This is because, in our proposed TDPP algorithm, 
higher link capacity will increase the throughput of the network, since more quantum channels will be available to generate qubits. Because of this, the trace-distance values between adjacent nodes are minimal leading to higher throughput.
In this scenario, EFiRAP and REPS algorithms show poor performance in throughput due to the following reasons. 
In EFiRAP, during entanglement generation, a round of purification operations is performed over links, which requires more qubits to be generated, as well as a larger quantum memory.
In the REPS, when channel capacity increases, more entanglement is generated between the adjacent nodes. However, most of it is destroyed due to network congestion and low values of EPR pairs.

\emph{3) \textbf{$\alpha > \beta$:}} 
Figures \ref{fig:my_label_2}(d), (g), and (h) demonstrate that an increase in link capacity leads to higher throughput of quantum networks in the TDPP algorithm. This is because, our proposed TDPP algorithm operates on the principle of minimizing trace-distance values between adjacent nodes, ensuring that the generated qubits are highly entangled and possess high-fidelity values. In contrast, algorithms such as Q-PATH, EFiRAP, and REPS consume more qubits to achieve similar high-fidelity values.
However, channel utilization decreases significantly when quantum memory becomes a bottleneck, as fewer qubits are available for entanglement and purification operations, substantially reducing quantum throughput in the network. Increasing the quantum memory hosted by each quantum node leads to higher throughput, but the memory utilization will decrease faster.
The performance of REPS is the worst since, in this algorithm, purification occurs before path selection and end-to-end fidelity can barely be assured.

\section{Conclusions} \label{conclusion}
This paper addresses the challenge of maintaining entanglement fidelity and reliable information flow in distributed quantum networks. 
A quantitative approach was proposed that ensures reliable information flow and preservation in quantum networks with multiple S-D pairs.
An iterative trace-distance based purification approach named TDPP was introduced that incorporates quantum state calculations, dynamic measures, and node centrality considerations.
It aims to preserve information and maintain E2E fidelity. We assess entanglement fidelity and information preservation by performing a quantum operation over a channel. The proposed TDPP approach emphasizes the importance of E2E fidelity preservation and proposes measures to preserve information between sender-receiver pairs. The results demonstrate that the TDPP algorithm significantly improves fidelity and throughput compared to existing methods across various utility factors.

In the future, we aim to address the routing problem with fidelity limitations by employing the "on-demand generation" concept. 
Additionally, we plan to investigate the correlation between the fidelity value and the likelihood of qubit error.

\begin{appendices}
\section{} \label{appendix1}
\textbf{The factor of $\frac{1}{2}$ in \eqref{eq1} ensures that the maximum possible value of the trace distance is 1.}

\begin{proof}
    
Suppose in a quantum network, a quantum node at $u$ prepares a state either $\rho$ or $\sigma$, each with probability $\frac{1}{2}$, and sends it to an adjacent quantum node $v$. Node $v$ must then discriminate between the two states using a binary measurement.
Let's consider a positive operator valued measure (POVM) $\left\{\Lambda,1-\Lambda\right\}$, where the measurement corresponds to state $\rho$ if $\Lambda$ occurs, and to state $\sigma$ otherwise. 
Then the success probability $P_{succ}$ of correctly identifying the incoming state at the adjacent node can be defined as:
\begin{equation*}
\begin{split}
    P_{succ} & =\frac{1}{2}\left(tr\left[\Lambda\rho\right]+tr\left[\left(1-\Lambda\right)\sigma\right]\right), \\ & =\frac{1}{2}\left(1+tr\left[\Lambda\left(\rho-\sigma\right)\right]\right), \\ & =\frac{1}{2}\left(1+D_{u,v}\left(\rho,\sigma\right)\right).
\end{split}
\end{equation*}
Thus, the success probability $P_{succ}$ is directly related to the trace-distance $D_{u,v}$, which quantifies the distinguishability between the two quantum states $\rho$ and $\sigma$ at nodes $u$ and $v$. 
The factor of $\frac{1}{2}$ ensures that the maximum distinguishability is bounded by 1.
Without this factor, the trace of $\left| \rho - \sigma \right|$ could be as large as $2$ for maximally distinct quantum states, since density matrices are positive semi-definite and can have a trace equal to 1. 
For example,
consider two pure states $\rho=\left.\left|0\right.\right\rangle\left\langle\left.0\right|\right.$ and $\sigma=\left.\left|1\right.\right\rangle\left\langle\left.1\right|\right.$, which are perfectly distinguishable. The difference $\rho-\sigma$ results in a matrix with eigenvalues $+1$ and $-1$. 
The trace of the absolute difference would therefore be $2$.
Thus, the factor of $\frac{1}{2}$ normalizes the trace distance, ensuring that its values lie between $0$ and $1$.

\end{proof}
\end{appendices}

\bibliographystyle{IEEEtran}
\bibliography{tdpp.bib}


\begin{IEEEbiographynophoto}{Pankaj Kumar}
 earned his M.Tech. degree in computer science and engineering from the Indian Institute of Technology (ISM) in Dhanbad, India. He is currently pursuing a Ph.D. degree in computer science and information engineering (CSIE) at National Taiwan University of Science and Technology (NTUST), Taiwan. His research focuses on network design, quantum algorithms, quantum information theory, quantum communications, and QKD networks.
 \end{IEEEbiographynophoto}

\begin{IEEEbiographynophoto}{Binayak Kar}
is an Assistant Professor of computer science and information engineering at National Taiwan University of Science and Technology (NTUST), Taiwan. He received his Ph.D. degree in computer science and information engineering from the National Central University (NCU), Taiwan, in 2018. He was a post-doctoral research fellow in computer science at National Chiao Tung University (NCTU), Taiwan, from 2018 to 2019. His research interests include edge computing, cybersecurity, and quantum computing.
\end{IEEEbiographynophoto}

\begin{IEEEbiographynophoto}{Shan-Hsiang Shen}
received the M.S. degree
from National Chiao Tung University, Republic
of China, in 2004, and the Ph.D. degree from
the University of Wisconsin, USA, in 2014.
He is currently an Associate Professor with the
Computer Science and Information Engineering
Department, National Taiwan University of Science
and Technology, Taiwan. His main research interests
include software-defined networking, network function virtualization, network security, and cloud computing.    
\end{IEEEbiographynophoto}

\end{document}